%% file: quickXsort_arXiv.tex
\documentclass[11pt]{article}

\usepackage{a4wide}
\usepackage[noend,vlined,figure]{algorithm2e} 
\usepackage[utf8]{inputenc}
\usepackage{amsfonts}
\usepackage{amssymb}
\usepackage{amsmath}
\usepackage{amsthm}
\usepackage{multicol}
\usepackage{tikz}
\usepackage{prettyref}
\usepackage{pgfplots}
\usepackage{microtype}
\usepackage{enumerate}

\usepackage[top=1in, bottom=1.5in, left=1in, right=1in]{geometry}

\newrefformat{thm}{Thm.~\ref{#1}}
\newrefformat{lm}{Lem.~\ref{#1}}
\newrefformat{dfn}{Def.~\ref{#1}}
\newrefformat{cor}{Cor.~\ref{#1}} 
\newrefformat{prop}{Prop.~\ref{#1}}
\newrefformat{sec}{Sect.~\ref{#1}}
\newrefformat{kap}{Chap.~\ref{#1}}
\newrefformat{bsp}{Ex.~\ref{#1}}
\newrefformat{eq}{(\ref{#1})}
\newrefformat{tab}{Tab.~\ref{#1}}
\newrefformat{fig}{Fig.~\ref{#1}}
\newrefformat{app}{App.~\ref{#1}}

\newcommand{\QuickXsort}{{\sc QuickXsort}}
\newcommand{\QuickXYsort}{{\sc QuickXYsort}}
\newcommand{\N}{\mathbb{N}}
\newcommand{\Z}{\mathbb{Z}}

\newcommand{\R}{\mathbb{R}}

\newcommand{\Merge}{\mbox{$\mathit{merge}$}}
\newcommand{\Mergeinsertion}{\mbox{$\mathit{mergeinsertion}$}}
\newcommand{\Mergeinsertionrecursive}{\mbox{$\mathit{mergeinsertionrecursive}$}}

\newcommand{\Link}{\mbox{$\mathit{join}$}}

\newcommand{\Construct}{\mbox{$\mathit{construct}$}}
\newcommand{\Binaryinsert}{\mbox{$\mathit{binary}${\rm -}$\mathit{insert}$}}

\newcommand{\Dancestor}{\mbox{$\mathit{d}$\mbox{\rm -}$\mathit{ancestor}$}}
\newcommand{\Dchild}{\mbox{$\mathit{d}$\mbox{\rm -}$\mathit{child}$}}

\newcommand{\Oh}{\mathcal{O}}
\newcommand{\oh}{o}

\newcommand{\floor}[1]{\left\lfloor\mathinner{#1} \right\rfloor}
\newcommand{\ceil}[1]{\left\lceil\mathinner{#1} \right\rceil}
\newcommand{\set}[2]{\left\{\, \mathinner{#1}\vphantom{#2}\: \left|\: \vphantom{#1}\mathinner{#2} \right.\,\right\}}

\renewcommand{\Pr}[1]{\mathop{\mathrm{Pr}}\left[\,#1\,\right]}
\providecommand{\DontPrintSemicolon}{\dontprintsemicolon}

\newtheorem{theorem}{Theorem}
\newtheorem{corollary}{Corollary}
\newtheorem{proposition}{Proposition}
\newtheorem{lemma}{Lemma}

\begin{document}

\title{QuickXsort: Efficient Sorting with 
       $n \log n - 1.399n +o(n)$ Comparisons on Average}
 \author{
   Stefan Edelkamp \\ 
   TZI, Universit{\"a}t Bremen, 
   Am Fallturm 1,  
   D-28239 Bremen, Germany\\ 
   \texttt{edelkamp@tzi.de}
\and
   Armin Wei\ss \\
   FMI, Universit{\"a}t Stuttgart, 
   Universit{\"a}tsstr.\ 38, 
   D-70569 Stuttgart, Germany\\ 
  \texttt{armin.weiss@fmi.uni-stuttgart.de}
}

\date{}

\maketitle
\thispagestyle{empty}

\begin{abstract}
In this paper we generalize the idea of {\sc QuickHeapsort} leading to
the notion of \QuickXsort. Given some external sorting algorithm~X,
\QuickXsort{} yields an internal sorting algorithm if~X
satisfies certain natural conditions.  

With 
{\sc QuickWeakHeapsort} and 
{\sc QuickMergesort} we present two examples for the \QuickXsort-construction.
Both are efficient algorithms that incur
approximately $n \log n - 1.26n +o(n)$ comparisons on the average. A worst
case of $n \log n + \Oh(n)$ comparisons can be achieved without significantly affecting
the average case.

Furthermore, we describe an implementation of {\sc
  MergeInsertion} for small $n$. Taking {\sc MergeInsertion}
as a base case for {\sc QuickMergesort}, we establish a worst-case
efficient sorting algorithm calling for $n \log n - 1.3999n + o(n)$
comparisons on average. 
{\sc QuickMergesort}
with constant size base cases shows the best performance on practical
inputs: when sorting integers it is slower by only 15\% to STL-{\sc Introsort}. 

\end{abstract}


\newpage

\setcounter{page}{1}
\section{Introduction}

Sorting a sequence of $n$ elements remains one of the most frequent tasks carried out by computers. A lower bound for sorting 
by only pairwise comparisons is
$\floor{\log n!} \approx n \log n - 1.44n + \Oh(\log n)$
comparisons for the worst
and average case (logarithms referred to by $\log$ are
base 2, 
the average case refers to a uniform distribution of all input permutations assuming all elements are different).

Sorting algorithms that are optimal
in the leading term are called \emph{constant-factor-optimal}. 
Table~\ref{compare} lists some milestones in the
race for reducing the coefficient in the linear term.
One of the most efficient (in terms of number of comparisons) {constant-factor-optimal} 
algorithms for solving the sorting problem is Ford
and Johnson's {\sc MergeInsertion} algorithm \cite{FordJ59}. It requires 
$n \log n - 1.329n + \Oh(\log n)$
comparisons in the worst case~\cite{Knu73}. {\sc MergeInsertion} has a severe
drawback that makes it uninteresting for practical issues: similar to
{\sc Insertionsort} the number of element moves is quadratic in $n$.
With {\sc Insertionsort} we mean the algorithm that inserts all elements
successively into the already ordered sequence finding the position for each
element by binary search (\emph{not} by linear search as mostly done).
However, {\sc MergeInsertion} and  {\sc Insertionsort} can be used to sort small
subarrays such that the quadratic running time for these subarrays is small in
comparison to the overall running time. 

Reinhardt~\cite{Reinhardt92} used this
technique to design an internal {\sc Mergesort} variant that needs in the worst
case $n \log n - 1.329n+ \Oh(\log n)$ comparisons.  Unfortunately, implementations of this {\sc InPlaceMergesort} algorithm have not been
documented. Katajainen
et al.'s~\cite{KatajainenPT96,ElmasryKS12} work
inspired by Reinhardt is practical, but the
number of comparisons is larger. 

Throughout the
text we avoid the terms \emph{in-place} or \emph{in-situ} and prefer
the term \emph{internal} (opposed to \emph{external}). We call an algorithm \emph{internal}
 if it needs at most $\Oh(\log n)$ space in addition to
the array to be sorted. That means we consider {\sc Quicksort} as an internal
algorithm whereas standard {\sc Mergesort} is external because it needs a linear
amount of extra space.

Based on {\sc QuickHeapsort} \cite{quickheap}, in this paper we
develop the concept of {\sc QuickXsort} and apply it to other sorting
algorithms as {\sc Mergesort} or {\sc WeakHeapsort}. This yields efficient internal sorting algorithms. The idea is very
simple: as in {\sc Quicksort} the array is partitioned into the
elements greater and less than some pivot element. Then one part of
the array is sorted by some algorithm X and the other part is sorted
recursively. The advantage of this procedure is that, if X is an external 
algorithm, then in {\sc
  QuickXsort} the part of the array which is not currently being sorted may be
used as temporary
space, what yields an internal variant of X. We show that under natural
assumptions {\sc
  QuickXsort} performs up to $\oh(n)$ terms on average the same number of
comparisons as X.

\begin{table}[t]
\vspace{-0.2cm}
\caption{Constant-factor-optimal sorting with $n \log n + \kappa n + o(n)$ comparisons.}\label{compare}
\vspace{-0.2cm}
\begin{center}
\begin{small}
\begin{tabular}{c|cc|cc|c} 
                                 &  Mem. & Other & {$\kappa$} Worst & {$\kappa$} Avg. & {$\kappa$} Exper. \\ \hline 
Lower bound                    & $\Oh(1)$ & $\Oh(n \log n)$  & -1.44 & -1.44 \\ \hline
{\sc BottomUpHeapsort} \cite{Weg93} & $\Oh(1)$ & $\Oh(n \log n)$ & $\omega(1)$ & -- &  [0.35,0.39] \\
{\sc WeakHeapsort} \cite{Dut93,EW00} & $\Oh(n/w)$  & $\Oh(n \log n)$ &  0.09 & -- & [-0.46,-0.42] \\
{\sc RelaxedWeakHeapsort} \cite{edelkampstiegeler} 
& $\Oh(n)$  & $\Oh(n \log n)$ &  -0.91 & -0.91 & -0.91 \\
{\sc Mergesort}~\cite{Knu73}   & $\Oh(n)$ & $\Oh(n \log n)$ & -0.91 & -1.26 & -- \\
{\sc ExternalWeakHeapsort} \# & $\Oh(n)$ & $\Oh(n \log n)$  &  -0.91 & -1.26* & -- \\ 
{\sc Insertionsort}~\cite{Knu73}  & $\Oh(1)$ & $\Oh(n^2)$ & -0.91 & -1.38 \#& -- \\
{\sc MergeInsertion}~\cite{Knu73}   & $\Oh(n)$ & $\Oh(n^2)$ &  -1.32 & -1.3999 \# & [-1.43,-1.41] \\
{\sc InPlaceMergesort} \cite{Reinhardt92}  & $\Oh(1)$ & $\Oh(n \log n)$ & -1.32 & -- & -- \\
{\sc QuickHeapsort}~\cite{quickheap,DiekertW13Quick} 
& $\Oh(1)$ & $\Oh(n \log n)$ & $\omega(1)$ & -0.03 &  $\approx$ 0.20 \\
& $\Oh(n/w)$ & $\Oh(n \log n)$ & $\omega(1)$ & -0.99 & $\approx$ -1.24 \\ \hline
{\sc QuickMergesort} (IS) \#  & $\Oh(\log n)$ & $\Oh(n \log n)$  &  -0.32  & -1.38 & -- \\
{\sc QuickMergesort} \#  & $\Oh(1)$ & $\Oh(n \log n)$  & -0.32  & -1.26 & [-1.29,-1.27] \\ 
{\sc QuickMergesort} (MI) \# ~ & $\Oh(\log n)$ & $\Oh(n \log n)$  &  -0.32  & -1.3999 & [-1.41,-1.40]\\ \hline
\end{tabular}
\end{small}
 \end{center}
%
Abbreviations: \# in this paper, MI MergeInsertion, -- not analyzed, * for $n=2^k$, 
$w$: computer word width in bits; 
we assume $\log n \in
  \Oh(n/w)$.

 For {\sc QuickXsort} we assume {\sc
     InPlaceMergesort} as a worst-case stopper (without  $ \kappa_{\mathrm{worst}} \in \omega(1)$).

\end{table}

The concept of {\sc QuickXsort} (without calling it like that) was
first applied in {\sc UltimateHeapsort} by Katajainen
\cite{Katajainen98}. In {\sc UltimateHeapsort}, first the
median of the array is determined, and then the array is
partitioned into subarrays of equal size. Finding the median means significant
additional effort. Cantone and Cincotti \cite{quickheap} weakened the
requirement for the pivot and designed {\sc QuickHeapsort} which uses
only a sample of smaller size to select the pivot for
partitioning. {\sc UltimateHeapsort} is inferior to {\sc
  QuickHeapsort} in terms of average case running time, although, unlike {\sc QuickHeapsort}, it
allows an $n\log n + \Oh(n)$ bound for the worst case number of comparisons. Diekert and 
Wei\ss~\cite{DiekertW13Quick} analyzed {\sc QuickHeapsort} more
thoroughly and showed that it needs less than $n\log n -0.99 n
+\oh(n)$ comparisons in the average case when implemented with
approximately $\sqrt{n}$ elements as sample for pivot selection and
some other improvements.

Edelkamp and Stiegeler \cite{edelkampstiegeler} applied the idea of {\sc QuickXsort} to {\sc WeakHeapsort} (which was first described by Dutton \cite{Dut93}) introducing {\sc QuickWeakHeapsort}. 
The worst case number
of comparisons of {\sc WeakHeapsort} is $n \lceil \log n \rceil -
2^{\lceil \log n \rceil} + n - 1 \le n \log n + 0.09n$, and, following
Edelkamp and Wegener \cite{EW00}, this bound is tight. 
In~\cite{edelkampstiegeler} an improved variant with $n \log n - 0.91n$
comparisons in the worst case and requiring extra space is
presented. With {\sc ExternalWeakHeapsort} we propose a further refinement 
with the same worst case bound, but in average requiring approximately
$n \log n - 1.26n$ comparisons. 
Using {\sc ExternalWeakHeapsort} as X in {\sc QuickXsort} we obtain an improvement over {\sc QuickWeakHeapsort} of \cite{edelkampstiegeler}.

As indicated above, {\sc Mergesort} is another good candidate to apply the {\sc QuickXsort}-construction. With {\sc QuickMergesort} we describe an internal variant of {\sc Mergesort} which not only in terms of number of comparisons is almost as good as {\sc Mergesort}, but also in terms of running time.
As mentioned before, {\sc MergeInsertion} can be used to sort small
subarrays. We study {\sc MergeInsertion} and provide an implementation
based on weak heaps. Furthermore, we give an average case analysis.  
When sorting small subarrays with
{\sc MergeInsertion}, we can show that the average number of
comparisons performed by {\sc
  Mergesort} is bounded by $n \log n - 1.3999n+ \oh(n)$, and, therefore,
{\sc QuickMergesort} uses at most $n \log n
- 1.3999n + \oh(n)$ comparisons in the average case.

\section{ \QuickXsort }\label{sec:quickXsort}

In this section we give a more precise description of {\sc QuickXsort} and derive some results concerning the number of comparisons performed in the average and worst case.
 Let X be some sorting
algorithm.
{\sc QuickXsort} works as follows:
First, choose some pivot element as median of some random sample. Next, partition the array according to this
pivot element, i.\,e., rearrange the array such that all elements left of the pivot are less or equal and all elements on the right are greater or equal than the pivot element.
Then, choose one part of
the array and sort it with algorithm X. (In general, it does not matter
whether the smaller or larger half of the array is chosen. However,
for a specific sorting algorithm X like Heapsort, there might be a
better and a worse choice.)  After one part of the array has been
sorted with X, move the pivot element to its correct position (right after/before the already sorted part) and sort the other
part of the array recursively with {\sc QuickXsort}.

The main advantage of this procedure is that the part of the array that  
 is not being sorted currently can be used as temporary memory for the
algorithm X. This yields fast \emph{internal} variants for
various \emph{external} sorting algorithms (such as {\sc Mergesort}). The idea
is that whenever a data element should be moved to the external storage, instead
it is swapped with some data element in the part of the array which is not
currently being sorted. Of course, this
works only, if the algorithm needs additional storage only for data
elements. Furthermore, the algorithm has to be able to keep track of the
positions of elements which have been swapped.  As the specific method depends on the algorithm X, we give some more details when we describe the examples for {\sc QuickXsort}.

For the number of comparisons we can derive some general results which hold for a wide class of algorithms X.
Under natural assumptions the average case number of comparisons of X and of \QuickXsort{} differs only by an
$o(n)$-term. For the rest of the paper, we assume
that the pivot is selected as the median of approximately $\sqrt{n}$
randomly chosen elements. 
Sample sizes of
approximately $\sqrt{n}$ are likely to be optimal as the results in
\cite{DiekertW13Quick,MartinezR01} suggest.

\begin{theorem} [\QuickXsort{}  Average-Case]\label{thm:quickXsort}
Let X be some sorting algorithm requiring at most $n \log n +
cn +o(n)$ comparisons in the average case. Then, \QuickXsort{}
implemented with $\Theta(\sqrt{n})$ elements as sample for pivot
selection is a sorting algorithm that also needs at most $n
\log n + cn +o(n)$ comparisons in the average case.
\end{theorem}

For the proofs we assume that the arrays are
indexed starting with $1$. The following lemma is crucial for our
estimates. It can be derived by applying Chernoff bounds or by direct elementary
calculations.
\begin{lemma}[{\cite[Lm.\ 2]{DiekertW13Quick}}]\label{lm:prob_bound}
Let $0<\delta < \frac{1}{2}$. If we choose the pivot as median of $2\gamma +1$ elements such that $2\gamma +1 \leq\frac{n}{2}$, then we have $\Pr{\text{pivot }\leq \frac{n}{2} - \delta n} < (2\gamma+1) \alpha^\gamma$ where $\alpha = 4\left(\frac{1}{4} - \delta^2\right)<1$.
\end{lemma}

\begin{proof}[Proof of \prettyref{thm:quickXsort}]
Let $T(n)$ denote the average number of comparisons performed by  \QuickXsort{}  on an input array of length $n$ and let $S(n)=n \log n + cn +s(n)$ with $s(n) \in o(n)$ be an upper bound for the average number of comparisons performed by the algorithm~X on an input array of length $n$. Without loss of generality we may assume that $s(n)$ is monotone.
We are going to show by induction that $$T(n) \leq n \log n + cn +t(n)$$
for some monotonically increasing $t(n) \in o(n)$ with $s(n) \leq
t(n)$ which we will specify later. %

Let $\delta(n) \in o(1) \cap \Omega(n^{-\frac{1}{4} + \epsilon})$ with
$1/n\leq \delta(n) \leq 1/4$, i.\,e., $\delta$ is some function tending slowly
to zero for $n\to \infty$. 
Because of $\delta(n) \in \Omega(n^{-\frac{1}{4} + \epsilon})$, we see that
$(2\gamma+1) \left(1 - 4\delta^2\right)^\gamma$ tends to zero if $\gamma \in
\Theta(\sqrt{n})$.
Hence, by \prettyref{lm:prob_bound} it follows that the probability that the
pivot is more than $n \cdot \delta(n)$ off the median $p(n) = \Pr{\text{pivot }
< n\left(\frac{1}{2} -\delta(n)\right)}+\Pr{\text{pivot }> n\left(\frac{1}{2}
+\delta(n)\right)}$ tends to zero for $n\to \infty$. 
In the following we write $ M= \left[n\left(\frac{1}{2}
-\delta(n)\right),n\left(\frac{1}{2} +\delta(n)\right)\right] \cap \N$ and
$\overline M = \{1,\dots,n\}\setminus M$. We obtain the following recurrence
relation:
\begin{align*}
T(n) &\leq n - 1 + T_{\mathrm{pivot}}(n)\\ &\qquad + \sum_{k=1}^n \Pr{\!\text{pivot }= k\!}\cdot\max\left\{\, T(k-1) + S(n-k), T(n-k)+S(k-1)\,\right\}\\
 &\leq n - 1 + T_{\mathrm{pivot}}(n) \\
 & \qquad+ \Pr{\!\text{pivot }\in M \!}\cdot\max_{k\in M}\left\{\, T(k) + S(n-k), T(n-k)+S(k)\,\right\}\\
  &\qquad + \Pr{\!\text{pivot }\in \overline M \!}\cdot\max_{k\in\overline
M}\left\{\, T(k) + S(n-k), T(n-k)+S(k)\,\right\}.
\end{align*}
The function $f(x) = x \log x + (n-x)\log(n-x)$, $f(0) = f(n) = n\log n$ has its
only minimum in the interval $[0,n]$ at $x = n/2$, i.\,e., for $0 < x < n/2$ it
decreases monotonically and for $n/2 < x < n$ it increases monotonically. We set
$\beta = \left(\frac{n}{2} + n\cdot\delta(n) \right)$. That means that we have
$f(x) \leq f (\beta)$ for $x\in M$ and  $f(x) \leq f(n)$ for $x\in
\overline{M}$. 
Using this observation, the induction hypothesis, and our assumptions, we
conclude
\begin{align*}
&\max_{k\in M}\left\{\, T(k) + S(n-k), T(n-k)+S(k)\,\right\}\\
&\qquad\leq \max\set{ f(k)+  cn + t(k) + s(n-k)}{k\in M} \leq f (\beta) + cn +
t(\beta) + s(n),\\\\
&\max_{k\in\overline M}\left\{\, T(k) + S(n - k), T(n-k)+S(k)\,\right\}\\
&\qquad\leq \max\set{f(k) +  cn + t(k) + s(n-k) }{k\in \overline
M}
\leq T(n) + s(n).
\end{align*}
With $p(n)$ as above we obtain:
\begin{align}
\begin{split}\label{eq:tn}
T(n)
 &\leq n - 1 + T_{\mathrm{pivot}}(n) + p(n) \cdot T(n) + s(n)\\
 &\qquad + (1-p(n)) \bigl(f(\beta) + cn + t(\beta)\bigr).
 \end{split}
\end{align}
We subtract $p(n) \cdot T(n)$ on both sides and then divide by $1- p(n)$. Let
$D$ be some constant such that $D \geq 1/(1-p(n))$ for all $n$ (which exists
since $p(n)\neq 1$ for all $n$ and $p(n)\to 0$ for $n\to \infty$). Then, we
obtain
\begin{align*}
T(n)
&\leq \left(1 + D\cdot p(n)\right)\cdot (n - 1) + D\cdot( T_{\mathrm{pivot}}(n) + s(n)) +  \bigl(f(\beta) + cn + t(\beta)\bigr)\\
 &\leq \left(1 + D\cdot p(n)\right)\cdot(n - 1) + D\cdot( T_{\mathrm{pivot}}(n) + s(n)) \\
 &\qquad + \left(\frac{n}{2} - n\cdot\delta(n) \right)\cdot(\log (n/2) + \log(1+2\delta(n)))\\
 &\qquad + \left(\frac{n}{2} + n\cdot\delta(n) \right) \cdot(\log (n/2) + \log(1+2\delta(n))) + cn + t(3n/4)\\
 &\leq
 n\log n + cn\\
  &\qquad +  \left(D\cdot p(n) + 2\cdot\delta(n)/\ln 2\right) \cdot n + D\cdot (
T_{\mathrm{pivot}}(n) +s(n)) + t(3n/4),
\end{align*}
where the last inequality follows from $ \log(1+x) = \ln(1+x)/\ln(2) \leq 
x/\ln(2)$ for $x\in  \R_{> 0}$.  We see that $T(n) \leq n \log n + cn +t(n)$
if $t(n)$ satisfies the inequality
\[\left(D\cdot p(n) + 2\cdot \delta(n)/\ln 2\right) \cdot n +
D\cdot\left(T_{\mathrm{pivot}}(n) + s(n)\right) + t(3n/4)\leq t(n).\]
We choose $t(n)$ as small as possible. Inductively, we can show that for every
$\epsilon$ there is some $D_\epsilon$ such that $t(n) < \epsilon n +
D_\epsilon$. Hence, the theorem follows.
\end{proof}

Does \QuickXsort{} provide a
good bound for the worst case? The obvious answer is ``no''. If always
the $\sqrt{n}$ smallest elements are chosen for pivot selection, a
running time of $\Theta(n^{3/2})$ is obtained. However, we can
prove that such a worst case is very unlikely. In fact, let
$R(n)$ be the worst case number of comparisons of the algorithm X.
\prettyref{prop:worstunlikely} states that the probability that
\QuickXsort{} needs more than $R(n) + 6n$ comparisons decreases
exponentially in $n$. (This bound is not tight, but since we do not aim for exact
probabilities, \prettyref{prop:worstunlikely} is enough for us.)

\begin{proposition}\label{prop:worstunlikely}
Let $\epsilon > 0$. The probability that \QuickXsort{} needs more than
$R(n) + 6n$ comparisons is less than $(3/4+ \epsilon)^{\sqrt[4]{n}}$
for $n$ large enough.
\end{proposition}

\begin{proof}
Let $n$ be the size of the input. We say that we are in a \emph{good} case if an
array of size $m$ is partitioned in the interval $[m/4,\, 3m/4]$, i.\,e., if the
pivot is chosen in that interval. We can obtain a bound for the desired
probability by estimating the probability that we always are in such a good case
until the array contains only $\sqrt{n}$ elements. For smaller arrays, we can
assume an upper bound of $\sqrt{n}^2 = n$ comparisons for the worst case.
For all partitioning steps that sums up to less than $n\cdot\sum_{i \geq
0}(3/4)^i =
4n$ comparisons if we are always in a good case. We also have to consider the
number of comparisons required to find the pivot element. At any stage the
pivot is chosen as median of at most $\sqrt{n}$ elements. Since the median can
be determined in linear time, for all stages together this sums up to less than
$n$ comparisons if we are always in a good case and $n$ is large enough.
Finally, for all the sorting phases with
X we need at most $R(n)$ comparisons in total (that is only a rough upper bound
which can be improved as in the proof of \prettyref{thm:quickXsort}). Hence, we
need at most $R(n) + 6n$ comparisons if always a good case occurs.

Now, we only have to estimate the probability that always a good case occurs.
 By \prettyref{lm:prob_bound}, the probability for a good case in the first partitioning step is at least $1-d\cdot\sqrt{n}\cdot\left(3/4\right)^{\sqrt{n}}$ for some constant $d$. 
We have to choose $\log(\sqrt{n})/\log(3/4)<1.21\log n$ times a pivot in the interval $[m/4,\, 3m/4]$, then the array has size less than $\sqrt{n}$. We  only have to consider partitioning steps where the array has size greater than $\sqrt{n}$ (if the size of the array is already less than $\sqrt{n}$ we define the probability of a good case as $1$). Hence, for each of these partitioning steps we obtain that the probability for a good case is greater than $1-d\cdot\sqrt[4]{n}\cdot\left(3/4\right)^{\sqrt[4]{n}}$. 
Therefore, we obtain
\begin{align*}\Pr{\text{always good case}}&\geq\left(1-d\cdot\sqrt[4]{n}\cdot\left(3/4\right)^{\sqrt[4]{n}}\right)^{1.21\log(n)}\\
&\geq 1 - 1.21\log(n) \cdot d\cdot\sqrt[4]{n}\cdot\left(3/4\right)^{\sqrt[4]{n}}
\end{align*}
by Bernoulli's inequality. For $n$ large enough we have $1.21\log(n) \cdot
d\cdot\sqrt[4]{n}\cdot\left(3/4\right)^{\sqrt[4]{n}}\leq (3/4+
\epsilon)^{\sqrt[4] {n}}$.
\end{proof}
To obtain a provable bound for the worst case complexity we
apply a simple trick.  We fix some worst case efficient sorting
algorithm Y. This might be, e.\,g., {\sc InPlaceMergesort}. Worst case efficient means that we
have a $n\log n + \Oh(n)$ bound for the worst case number of
comparisons.  We choose some slowly decreasing function $\delta(n) \in
o(1) \cap \Omega(n^{-\frac{1}{4} + \epsilon})$, e.\,g., $\delta(n) =
1/ \log n$. Now, whenever the pivot is more than $n\cdot\delta(n)$
off the median, we switch to the algorithm Y. We call this \QuickXYsort{}.
To achieve a good worst case bound, of course, we also need a
good bound for algorithm X. W.\,l.\,o.\,g.\ we assume the
same worst case bounds for X as for Y.  Note that
\QuickXYsort{} only makes sense if one needs a provably good worst case 
bound. Since \QuickXsort{} is always expected to make at most as
many comparisons as \QuickXYsort{} (under the reasonable assumption
that X on average is faster than Y -- otherwise one would use simply Y), in every step of the recursion
\QuickXsort{} is the better choice for the average case.

In order to obtain an efficient internal sorting algorithm, of course, Y has to be internal and X using at most $n$ extra spaces for an array of size $n$.

\begin{theorem}[\QuickXYsort{} Worst-Case]\label{thm:QuickXYsort}
Let X be a sorting algorithm with at most $n \log n + cn +o(n)$
comparisons in the average case and $R(n) = n \log n + dn + o(n)$
comparisons in the worst case ($d \geq c$). Let Y be a sorting
algorithm with at most $R(n)$ comparisons in the worst case.  Then,
\QuickXYsort{} is a sorting algorithm that performs at most
$n \log n + cn + o(n)$ comparisons in the average case and $n \log n +
(d+1) n + o(n)$ comparisons in the worst case.
\end{theorem}

\begin{proof}
Since the proof is very similar to the proof of \prettyref{thm:quickXsort}, we provide only a sketch. By replacing $T(n)$ by $R(n) = n \log n + dn + r(n)$ with $r(n) \in o(n)$ in the right side of \prettyref{eq:tn} in the proof of \prettyref{thm:quickXsort} we obtain for the average case:
\begin{align*}
T_{\mathrm{av}}(n)
 &\leq (n - 1) + T_{\mathrm{pivot}}(n)+ s(n) + p(n)\cdot( n\log n + dn + r(n) )\\
 &\qquad  + (1-p(n))\cdot\left( \left(\frac{n}{2} -n\cdot\delta(n)\right)\cdot(\log( n/2) + \log(1+2\delta(n)) + c)  \right.\\
 &\qquad\qquad \left. + \left(\frac{n}{2} + n\cdot\delta(n) \right)\cdot(\log (n/2) + \log(1+2\delta(n)) + c) + t(3n/4)\right)\\
  &\leq n\log n + cn + T_{\mathrm{pivot}}(n)\\
  &\qquad + \left(p(n) \cdot(dn+r(n)) + 2\delta(n)/\ln 2\right) \cdot n +  s(n)
+ t(3n/4).
\end{align*}
As in \prettyref{thm:quickXsort} the statement for the average case follows.

For the worst case, there are two possibilities: either the algorithm already
fails the condition $ \mathrm{pivot} \in \left[n\left(\frac{1}{2}
-\delta(n)\right),\, n\left(\frac{1}{2} +\delta(n)\right)\right]$ in the first
partitioning step or it does not. In the first case, it is immediate that we
have a worst case bound of $n \log n + dn + n + o(n)$, which also is tight. Note
that we assume that we can choose the pivot element in time $o(n)$ which is no
real restriction, since the median of $\Theta(\sqrt{n})$ elements can be found
in $\Theta(\sqrt{n})$ time. In the second case, we assume by induction that
$T_{\mathrm{worst}}(m) \leq m \log m + dm + m + u(m)$ for $m< n$ for some
$u(m)\in o(m)$ and obtain a recurrence relation similar to \prettyref{eq:tn} in
the proof of \prettyref{thm:quickXsort}:
\begin{align*}
\begin{split}\label{eq:tn}
T_{\mathrm{worst}}(n)
 &\leq n - 1 + T_{\mathrm{pivot}}(n)  +  R\left(\frac{n}{2}
-n\cdot\delta(n)\right) + T_{\mathrm{worst}}\left( \frac{n}{2}
+n\cdot\delta(n)\right).
 \end{split}
\end{align*}
By the same arguments as above the result follows.
\end{proof}

\section{{\sc QuickWeakHeapsort}}\label{sec:QWHS}

In this section consider {\sc QuickWeakHeapsort} as a first example of \QuickXsort{}. We start by introducing weak heaps and then continue by describing {\sc WeakHeapsort} and a novel external
version of it. This external version is a good candidate for \QuickXsort{} and yields an efficient sorting algorithm that uses approximately
$n\log n - 1.2n$ comparisons (this value is only a rough estimate and neither a bound from below nor above). A drawback of  {\sc WeakHeapsort} and its variants is that they require one
extra bit per element. The exposition also serves as an  
intermediate step towards our implementation of {\sc MergeInsertion}, where the weak-heap data structure will
be used as a building block.

\begin{figure}[t!]
\begin{center}
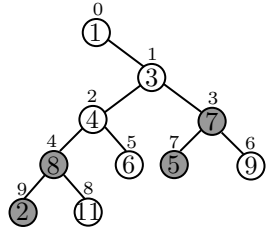
\end{center}
\caption{
  A weak heap 
  (reverse bits are set for
  grey nodes, above the nodes are array
  indices.)\label{fig:example}}
\end{figure}

Conceptually, a \emph{weak heap} (see Fig.~\ref{fig:example}) 
is a binary tree satisfying the following conditions:
\begin{enumerate}[(1)]
 \item The root of the entire tree has no left child.
 \item Except for the root, the nodes
that have at most one child are in the last two levels only. 
Leaves at the last level can be scattered, i.\,e., the last level is not
necessarily filled from left to right. 
\item Each node stores an element that is smaller than or equal to every 
element stored in its right subtree.
\end{enumerate}
From the first two properties we deduce that the height of a weak heap
that has $n$ elements is $\lceil \log n \rceil + 1$.  The third
property is called the \emph{weak-heap ordering} or {half-tree
  ordering}. In particular, this property enforces no relation between
an element in a node and those stored its left subtree. On the other hand, it implies that any node together with its right subtree forms a weak heap on its own.
In an array-based implementation, besides the element array $s$, an
array $r$ of \emph{reverse bits} is used, i.\,e., $r_i \in \{0,1\}$ for
$i \in \{0,\ldots,n-1\}$. The root has index $0$. The array index of the left child of $s_i$
is $2i + r_i$, the array index of the right child is $2i +1 - r_i$,
and the array index of the parent is $\lfloor i/2 \rfloor$ (assuming
that $i \neq 0$). Using the fact that the indices of the left and
right children of $s_i$ are exchanged when flipping $r_i$, subtrees can
be reversed in constant time by setting $r_i \leftarrow 1-r_i$.
The \emph{distinguished ancestor} ($\Dancestor{}(j)$) of $s_j$ for $j \neq 0$, is recursively defined as the parent
of $s_j$ if $s_j$ is a right child, and the distinguished ancestor of
the parent of $s_j$ if $s_j$ is a left child. 
The distinguished ancestor of  $s_j$ is the first element on the path from  $s_j$ to the root which is known to be smaller or equal than  $s_j$ by (3). Moreover, any subtree rooted by $s_j$, together with the distinguished ancestor $s_i$ of $s_j$, forms again a weak heap with root $s_i$ by considering $s_j$ as right child of $s_i$.

The basic operation for creating a weak heap is the \Link{} operation which combines two weak heaps into one.
Let $i$ and $j$ be two nodes in a weak
heap such that $s_i$ is smaller than or equal to every
element in the left subtree of $s_j$. Conceptually, $s_j$ and its
right subtree form a weak heap, while $s_i$ and the left subtree of
$s_j$ form another weak heap.  (Note that $s_i$ is not allowed be in the subtree with root $s_j$.) 
The result of \Link{} is a weak heap with root at position $i$.
 If $s_j<s_i$, the two elements are swapped and $r_j$ is flipped. As a result, the
new element $s_j$ will be smaller than or equal to every element in its
right subtree, and the new element $s_i$ will be smaller than or equal
to every element in the subtree rooted at $s_j$. To sum up, \Link{}
requires constant time and involves one element comparison and a
possible element swap in order to combine two weak heaps to a new one.

%

The construction of a weak heap consisting of $n$ elements requires $n-1$ comparisons.
In the standard bottom-up construction of a weak heap the nodes are
visited one by one. Starting with the last node in the array and moving to the front, the two weak heaps rooted at
a node and its distinguished ancestor are joined. 
The amortized cost to get from a node
to its distinguished ancestor is $\Oh(1)$~\cite{EW00}. 

When using weak heaps for sorting, the minimum is removed and the weak heap condition restored until the weak heap becomes empty.
After extracting an element from the root, 
first the \emph{special path} from the root is traversed top-down, and then, 
in a bottom-up process the weak-heap property 
is restored using at most $\lceil \log n \rceil$ 
join operations.
(The special path is established by going once to the
right and then to the left as far as it is possible.) 
Hence, extracting the minimum requires at most $\lceil \log n \rceil$ comparisons.

Now, we introduce a modification to the standard procedure described by Dutton \cite{Dut93}, which has a slightly improved performance, but requires  extra space.  We call this modified algorithm {\sc ExternalWeakHeapsort}. This is because it needs an extra output array, where the elements which are extracted from the weak heap are moved to.
On average {\sc ExternalWeakHeapsort} requires less comparisons than {\sc
  RelaxedWeakHeapsort}~\cite{edelkampstiegeler}.  Integrated in
\QuickXsort{} we can implement it without extra space other than
the extra bits $r$ and some other extra bits.
We introduce an additional array \emph{active} and weaken the
requirements of a weak heap: we also allow nodes on other than the
last two levels to have less than two children. Nodes where the
\emph{active} bit is set to false are considered to have been removed. {\sc
  ExternalWeakHeapsort} works as follows: First, a usual weak heap is
constructed using $n-1$ comparisons. Then, until the weak heap
becomes empty, the root~-- which is the minimal element~-- is moved to
the output array and the resulting hole has to be filled
with the minimum of the remaining elements (so far the only difference to normal {\sc WeakHeapsort} is that there is a separate output area). 

The hole is filled
by searching the special path from the root to a node $x$ which has no
left child. Note that the nodes on the special path are exactly the nodes having the root as distinguished ancestor. Finding the special path does not
need any comparisons, since one only has to follow the reverse
bits. Next, the element of the node $x$
is moved to the root leaving a hole. If $x$ has a right subtree (i.\,e., if $x$ is the root of a weak heap with more than one element), this
hole is filled by applying the hole-filling algorithm recursively to
the weak heap with root $x$. Otherwise, the \emph{active} bit of $x$ is set
to false. Now, the root of the whole weak heap together with the subtree rooted by $x$ forms a
weak heap. However, it remains to restore the weak heap condition for the whole
weak heap.  Except for the root and $x$, all nodes on the special path
together with their right subtrees form weak heaps. Following the
special path upwards these weak heaps are joined with their distinguished ancestor as during the weak heap construction (i.\,e., successively they are joined with the weak heap consisting of the root and the already treated nodes on the special path together with their subtrees). Once, all the weak heaps on the special path are
joined, the whole array forms a weak heap again.

\begin{theorem}\label{thm:whsms}
 For $n = 2^k$ {\sc ExternalWeakHeapsort}
performs exactly the same comparisons as {\sc Mergesort} applied on a
fixed permutation of the same input array.
\end{theorem}

\begin{proof}
 First, recall the {\sc
  Mergesort} algorithm: The left half and the right half of the array
are sorted recursively and then the two subarrays are merged together
by always comparing the smallest elements of both arrays and moving
the smaller one to the separate output area.
Now, we move to {\sc WeakHeapsort}. Consider the tree as it is initialized with all 
reverse bits set to \emph{false}. Let $r$ be the root and $y$
its only child (not the elements but the positions in the
tree). We call $r$ together with the left subtree of $y$ the left part
of the tree and we call $y$ together with its right subtree the right
part of the tree. That means the left part and the right part form
weak heaps on their own.
The \emph{only} time an element is moved from the right to the left
part or vice-versa is when the data elements $s_r$ and $s_y$ are
exchanged. 
However, always one of the data elements of $r$ and $y$ comes from the right part and one from the left part. After extracting 
the minimum $s_r$, it is replaced by the smallest remaining element of the part $s_r$ came from. Then, the new $s_r$ and $s_y$ are compared again and so on.
Hence, for extracting the
elements in sorted order from the weak heap the following happens. First, the smallest
elements of the left and right part are determined, then they are
compared and finally the smaller one is moved to the output area. This
procedure repeats until the weak heap is empty. This is
exactly how the recursion of {\sc Mergesort} works: always the smallest
elements of the left and right part are compared and the smaller one
is moved to the output area. 
If $n=2^k$, then the left
and right parts for {\sc Mergesort} and {\sc WeakHeapsort} have the same
sizes. 
\end{proof}

By \cite[5.2.4--13]{Knu73} we obtain the following corollary.

\begin{corollary}[Average Case {\sc ExternalWeakHeapsort}]\label{cor:ewhs}
 For $n = 2^k$ the algorithm {\sc ExternalWeakHeapsort} uses
 approximately $n \log n - 1.26n$ comparisons in the average case.
\end{corollary}

If $n$ is not a power of two, the sizes of left and right parts
of {\sc WeakHeapsort} are less balanced than the left and right parts of
ordinary {\sc Mergesort} and one can expect a slightly higher number of
comparisons. 
For {\sc QuickWeakHeapsort}, the half of the array which is not sorted by {\sc ExternalWeakHeapsort} is used as output area.
Whenever the root is moved to the output area, the element that occupied that place before is inserted as a dummy element at the position where the \emph{active} bit is set to false. 
Applying \prettyref{thm:quickXsort}, we obtain the rough estimate of $n \log n - 1.2n$ comparisons for the average case of {\sc QuickWeakHeapsort}.

\section{{\sc QuickMergesort}}\label{sec:QMS}

As another example for \QuickXsort{} we consider {\sc
  QuickMergesort}. For the {\sc Mergesort} part we use standard
(top-down) {\sc Mergesort} which can be implemented using $m$ extra
spaces to merge two arrays of length $m$ (there are other methods like
in \cite{Reinhardt92} which require less space~-- but for our purposes
this is good enough). The procedure is depicted in
\prettyref{fig:merge}. We sort the larger half of the partitioned
array with {\sc Mergesort} as long as we have one third of the whole
array as temporary memory left, otherwise we sort the smaller part with {\sc Mergesort}. 
\begin{figure}[t]
\begin{multicols}{2}
  \begin{center}
  \begin{scriptsize}
\begin{tikzpicture}[scale = 0.5]

\draw(0,0) -- (9,0);
\draw(0,1) -- (9,1);
\draw(0,0) -- (0,1);

\draw[thick](6,0) -- (6,1);
 \node[] (st) at (6,-0.3) {Pivot};
 
\draw[->](4.5,1) ..controls(4.5,2) and (7.5, 2 ).. (7.5,1);

\draw(3,0) -- (3,1);
\draw(9,0) -- (9,1);
\end{tikzpicture}
\end{scriptsize}
\end{center}

  \begin{center}
  \begin{scriptsize}
\begin{tikzpicture}[scale = 0.5]

\draw(0,0) -- (9,0);
\draw(0,1) -- (6,1);
\draw(6,0.3) -- (9,1);
\draw(0,0) -- (0,1);

\draw[thick](6,0) -- (6,1);
 \node[] (st) at (6,-0.3) {Pivot};
 
\draw[->](1.5,1) ..controls(1.5,2) and (4.5, 2 ).. (4.5,1);

\draw(3,0) -- (3,1);
\draw(9,0) -- (9,1);
\end{tikzpicture}
\end{scriptsize}
\end{center}
\end{multicols}

  \begin{center}
  \begin{scriptsize}
\begin{tikzpicture}[scale = 0.5]

\draw(0,0) -- (9,0);
\draw(0,1) -- (3,1);
\draw(3,0.3) -- (6,1);
\draw(6,0.3) -- (9,1);
\draw(0,0) -- (0,1);

\draw[thick](6,0) -- (6,1);
 \node[] (st) at (6,-0.3) {Pivot};
 
\draw[->](4.5,0.65) ..controls(4,1.7) and (1.5, 1.7 ).. (1,1);
\draw[->](7.5,0.65) ..controls(7,2) and (1.5, 2 ).. (1,1);

\draw(3,0) -- (3,1);
\draw(9,0) -- (9,1);
\end{tikzpicture}
\end{scriptsize}
\end{center}

\vspace{-0.5cm}
\caption{First the two halves of the left part are sorted moving them from
  one place to another. Then, they are merged to the original place.}\label{fig:merge}
\end{figure}
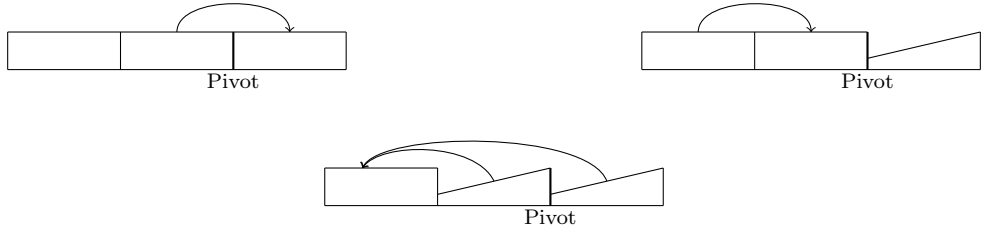
Hence, the part which is not sorted by {\sc Mergesort} always provides enough temporary space. When a data element should be moved to or from the temporary space, it is swapped with the element occupying the respective position. Since {\sc Mergesort} moves through the data from left to right, it is always known which are the elements to be sorted and which are the dummy elements.
Depending on the
implementation the extra space needed is $\Oh(\log n)$ words for
the recursion stack of {\sc Mergesort}. By avoiding recursion this
can even be reduced to $\Oh(1)$. \prettyref{thm:quickXsort} together with \cite[5.2.4--13]{Knu73} yields the next
result.

\begin{theorem} [Average Case {\sc QuickMergesort}]\label{thm:QMS}
{\sc QuickMergesort} is an internal sorting algorithm that performs at most $n \log n - 1.26n + o(n)$
comparisons on average.
\end{theorem}

We can do even better if we sort small subarrays with another
algorithm Z requiring less comparisons but extra space and more moves, e.\,g., {\sc Insertionsort}
 or {\sc MergeInsertion}.  If we use $\Oh(\log n)$
elements for the base case of {\sc Mergesort}, we have to call Z at most
$\Oh(n/ \log n)$ times. In this case we can allow additional
operations of Z like moves in the order of $\Oh(n^2)$, given that $\Oh((n/
\log n)\cdot \log^2 n ) = \Oh(n \log n)$.

Note that for
the next theorem we only need that the size of the base cases grows as $n$ grows. 
Nevertheless, $\Oh(\log n)$
is the largest growing value we can choose if we apply a base case 
algorithm with
$\Theta(n^2)$ moves and want to achieve an $\Oh(n\log n)$ overall
running time.

\begin{theorem}[{\sc QuickMergesort} with Base Case]\label{thm:QMSbase}
Let $Z$ be some  sorting algorithm with
$n \log n +en + o(n)$ comparisons on the average and other operations
taking at most $\Oh(n^2)$ time. If base cases of size $\Oh(\log n)$ are sorted with Z, {\sc QuickMergesort} uses at most
 $n \log n +en + o(n)$ comparisons and $\Oh(n\log n)$ other instructions on the average.
\end{theorem}

\begin{proof}
 By \prettyref{thm:quickXsort} and the preceding remark, the only thing we have to prove is
 that {\sc Mergesort} with base case Z requires on average at most
 $\leq n \log n + en + o(n)$ comparisons, given that Z needs
 $\leq U(n) = n \log n + en + o(n)$ comparisons on average.
 The latter means that for every $\epsilon >0 $ we have $U(n) \leq
 n\log n + (e +\epsilon)\cdot n$ for $n$ large enough.  
 
 Let $S_{k}(m)$
 denote the average case number of comparisons of {\sc Mergesort} with
 base cases of size $k$ sorted with Z and let $\epsilon > 0$. Since
 $\log n$ grows as $n$ grows, we have that $S_{\log n}(m) = U(m) \leq
 m\log m + (e +\epsilon)\cdot m$ for $n$ large enough and $(\log n )/2
 < m \leq \log n$.
For $m > \log n$ we have $S_{\log n}(m) \leq 2 \cdot S_{\log n}(m/2) + m$
and by induction we see that $S_{\log n}(m) \leq m\log m + (e +\epsilon)\cdot m$. Hence, also $S_{\log n}(n) \leq n\log n + (e +\epsilon)\cdot n$ for $n$ large enough.

\end{proof}

Using {\sc Insertionsort} we obtain the following result. Here, $\ln$ denotes the natural logarithm.
As we did not find a result in literature, we also provide a proof. 
Recall that {\sc Insertionsort} inserts the elements one by one into the already sorted sequence by binary search.

\begin{proposition}[Average Case of {\sc Insertionsort}]\label{prop:avgIns}
The sorting algorithm {\sc Insertionsort} needs $n\log n - 2 \ln
2 \cdot n + c(n)\cdot n + \Oh(\log n)$ comparisons on the average where
$c(n) \in [-0.005,0.005]$.
\end{proposition}

\begin{corollary}[{\sc QuickMergesort} with Base Case {\sc Insertionsort}]
If we use as base case {\sc Insertionsort}, {\sc
  QuickMergesort} uses at most $n \log n - 1.38 n + o(n)$ comparisons
and $\Oh(n \log n)$ other instructions on the average.
\end{corollary}

\begin{proof}[Proof of \prettyref{prop:avgIns}]
 First, we take a look at the average number of comparisons  $T_{\mathrm{InsAvg}}(k)$ to insert
 one element into a sorted array of $k-1$ elements by binary insertion.
 
 To insert a new element into $k-1$ elements either needs $\ceil{\log
   k}-1$ or $\ceil{\log k}$ comparisons. There are $k$ positions where
 the element to be inserted can end up, each of which is equally likely. For $2^{\ceil{\log k}} - k$ of
 these positions $\ceil{\log k}-1$ comparisons are needed. For the
 other $k - (2^{\ceil{\log k} } - k ) = 2k -2^{\ceil{\log k}}$
 positions $\ceil{\log k}$ comparisons are needed. This means
 \begin{eqnarray*}
 T_{\mathrm{InsAvg}}(k) &=&\frac{(2^{\ceil{\log k}} - k )\cdot (\ceil{\log{k}}-1) + (2k -2^{\ceil{\log k}}) \cdot \ceil{\log k} }{k}\\
  &=& \ceil{\log k } + 1 - \frac{2^{\ceil{\log k} }}{k}
 \end{eqnarray*}
 comparisons are needed on average. By \cite[5.3.1--(3)]{Knu73}, we obtain for the average case for sorting $n$ elements:
  \begin{eqnarray*}
T_{\mathrm{InsSortAvg}}(n) &=& \sum_{k=1}^{n}  T_{\mathrm{InsAvg}}(k)= \sum_{k=1}^{n}  \left(\ceil{\log k } + 1- \frac{2^{\ceil{\log k}}}{k}\right) \\
   &=&n\cdot\ceil{\log n} -  2^{\ceil{\log n}} + 1 + n - \sum_{k=1}^{n}  \frac{2^{\ceil{\log k}}}{k}.
 \end{eqnarray*}
 We examine the last sum separately. In the following we write $H(n) =
 \sum_{k = 1}^n\frac{1}{k} = \ln n + \gamma + \Oh(\frac{1}{n})$ for
 the harmonic sum with $\gamma$ the Euler
 constant.
  \begin{eqnarray*}
\sum_{k=1}^{n}  \frac{2^{\ceil{\log k}}}{k} 
  &=& 1 + \sum_{i=0}^{\ceil{\log n}-2} \sum_{\ell = 1}^{2^i}\  \frac{2^{i+1}}{2^i+\ell}\ + \sum_{\ell = 2^{\ceil{\log n}-1} + 1}^{n}  \frac{2^{\ceil{\log n}}}{\ell} \\
&=& 1+\left(\sum_{i=0}^{\ceil{\log n}-2} 2^{i+1}\cdot\Bigl(H\left(2^{i+1}\right) - H\left(2^i\right)\Bigr) \right) +  2^{\ceil{\log n}}\cdot\left(H(n) - H\bigl(2^{\ceil{\log n}-1}\bigr)\right) \\
&=& \sum_{i=0}^{\ceil{\log n}-2} 2^{i+1}\cdot\left(\ln \left(2^{i+1}\right)) + \gamma- \ln\left(2^i\right)-\gamma \right)\\ 
&& \qquad+ \left(\ln\left(2^{n}\right) + \gamma -\ln\bigl(2^{\ceil{\log n}-1}\bigr) - \gamma \right)\cdot 2^{\ceil{\log n}} + \Oh(1)\\
&=& \ln 2 \cdot \sum_{i=0}^{\ceil{\log n}-2} 2^{i+1}\cdot(i+1-i )\\ 
&& \qquad+ \Bigl(\log(n)\cdot \ln 2  -({\ceil{\log n}-1})\cdot \ln 2\Bigr)\cdot 2^{\ceil{\log n}}\\
&=& \ln 2 \cdot \left( 2\cdot \bigl(2^{\ceil{\log n}-1} -1\bigr) + (\log n -\ceil{\log n} + 1)\cdot 2^{\ceil{\log n}}\right)+ \Oh(1)\\ 
&=& \ln 2 \cdot \left( 2  + \log n -\ceil{\log n}\right)\cdot 2^{\ceil{\log n}}+ \Oh(1)\\ 
 \end{eqnarray*}
Hence, we have 
  \begin{eqnarray*}
T_{\mathrm{InsSortAvg}}(n)  
&=&n\cdot\ceil{\log n} -  2^{\ceil{\log n}}  + n- \ln 2 \cdot \left( 2  + \log n -\ceil{\log n}\right)\cdot 2^{\ceil{\log n}}+ \Oh(1).
 \end{eqnarray*}
 In order to obtain a numeric bound for $T_{\mathrm{InsSortAvg}}(n)$,
 we compute $(T_{\mathrm{InsSortAvg}}(n) - n\log n)/n$ and then
 replace $\ceil{\log n} - \log n$ by $x$. This yields a function
\[x\mapsto x- 2^x + 1 - \ln 2\cdot(2 -x)\cdot 2^x,\]
which oscillates between $-1.381$ and $-1.389$ for $0\leq x < 1$. For
$x=0$, its value is $2 \ln 2 \approx 1.386$.
\end{proof}

Bases cases of growing size,
always lead to a constant factor overhead in running time if an algorithm
with a quadratic number of total operations is applied. Therefore, in the
experiments we will also consider constant size base cases which
offer a slightly worse bound for the number of comparisons, but are
faster in practice. We do not analyze them separately, since the
preferred choice for the size depends on the type of data to be sorted and the system on which the
algorithms run.

\section{{\sc MergeInsertion}}

{\sc MergeInsertion} by Ford and Johnson \cite{FordJ59} 
is one of the best sorting algorithms in terms of number of comparisons. Hence, it can be applied for sorting base cases of {\sc QuickMergesort} what yields even better results than {\sc Insertionsort}. Therefore, we want to give a brief description of the algorithm and our implementation.
While the description is simple, {\sc MergeInsertion}  is not easy to implement efficiently.
Our implementation is based on weak heaps and uses $n \log n + n$ extra bits. 
Algorithmically, {\sc MergeInsertion}$(s_0,\ldots,s_{n-1})$  
can be
described as follows (an intuitive example for $n=21$ can be found in~\cite{Knu73}).

\begin{enumerate}
\item Arrange the input such that $s_i \geq
  s_{i+\floor{n/2}}$ for $0 \leq i < \floor{n/2}$ with one comparison
  per pair. Let $a_i = s_i$ and $b_i = s_{i+\floor{n/2}} $ for
$0 \leq i < \floor{n/2}$, and $b_{\floor{n/2}} = s_{n-1}$ if $n$ is odd.
\item Sort the values $a_0{,}...{,}a_{\lfloor n/2 \rfloor-1}$
  recursively with {\sc MergeInsertion}.
\item 
Rename the solution as follows:  $b_0 \le a_0 \le a_1 \le \dots \le a_{\lfloor n/2
   \rfloor-1}$ and insert the elements  $b_1,\ldots,b_{\lceil n/2
  \rceil-1}$ via binary insertion, following the ordering $b_2$,
$b_1$, $b_4$, $b_3$, $b_{10}$, $b_9, \dots,b_5,\dots$, $b_{t_{k-1}}$, $b_{t_{k-1}-1}, \dots b_{t_{k-2}+1}$, $b_{t_{k}}, \dots$ 
into the main chain, where  $t_k = (2^{k+1}+(-1)^k)/3$.
\end{enumerate}

Due to the different renamings, the recursion, and the
change of link structure, the design of an efficient implementation
is not immediate.
Our proposed implementation of {\sc MergeInsertion} is based on a
tournament tree representation with weak heaps as in
\prettyref{sec:QWHS}. 
The pseudo-code implementations for all the operations to construct a
tournament tree with a weak heap and to access the partners in each
round are shown in Fig.~\ref{tournament} in the appendix. (Note that for simplicity in the above formulation the
indices and the order are reversed compared to our implementation.)

One main subroutine of {\sc MergeInsertion} is binary insertion. The call \Binaryinsert$(x,y,z)$ inserts
the element at position $z$ between position $x - 1$ and $x+y$ by
binary insertion. (The pseudo-code
implementations for the binary search routine is shown in
\prettyref{fig:binarysearch} in the appendix.)
In this routine we do not move the data elements themselves, but
we use an additional index array $\phi_0,\ldots,\phi_{n-1}$ to point to the
elements contained in the weak heap tournament tree and 
move these indirect addresses. This approach has the advantage that the relations stored in the tournament tree are preserved.

The most important procedure for {\sc MergeInsertion} is the organization
of the calls for \Binaryinsert{}. After adapting the addresses for the elements $b_i$
(w.\,r.\,t.\ the above description) in the second part of the array, the
algorithm calls the binary insertion routine with appropriate indices.
Note that we always use $k$ comparisons for all elements of the $k$-th
block (i.\,e., the elements $b_{t_k},\dots, b_{t_{k-1}+1}$) even if there might be the chance to save one comparison. 
By introducing an additional array, which for each $b_i$ contains the
current index of $a_i$, we can exploit the
observation that not always $k$ comparisons are needed to insert an
element of the $k$-th block. In the following we call this the
\emph{improved} variant.
The pseudo-code of the basic variant is shown in Fig.\ \ref{xxmerge}. The last sequence
is not complete and is thus tackled in a special case. 

\begin{algorithm}[t!]
\SetKwInput{Proc}{procedure}
\SetKwInput{Global}{global}
\SetKwFor{For}{for}{}{}
\SetKwFor{While}{while}{}{}
\SetKw{Return}{return}

\Proc{\Merge{}(%
  $m$: integer)
}
\Global{$\phi$ array of $n$ integers imposed by weak-heap}

    \For {$l \leftarrow 0\ \mbox{\bf to}\ \floor{m/2}-1$} {
      $\phi_{m - odd(m)-l-1} \leftarrow \Dchild(\phi_l,m - odd(m))$\;
    }
    $k \leftarrow 1; e \leftarrow  2^k;
    c \leftarrow f \leftarrow 0$\;
    \While {$e < m$} {
      $k \leftarrow k+1$;
      $e \leftarrow 2e$\;
      $l \leftarrow \ceil{m/2} + f$;
      $f \leftarrow f + (t_k-t_{k-1})$\;
      \For {$i \leftarrow 0\ \mbox{\bf to}\ (t_k-t_{k-1})-1$} {
        $c \leftarrow c + 1$\;
        \If{$c = \ceil{m/2}$} {
          \Return\;
        }
	\If {$t_{k} > \ceil{m/2}-1$} { 
	  \Binaryinsert($i+1-odd(m),l,m-1$)\;
        }
	\Else { 
	  \Binaryinsert($\floor{m/2}-f+i,e-1,\floor{m/2}+f)$\;
        }
      }
}
\caption{Merging step in {\sc MergeInsertion} with $t_k = (2^{k+1}+(-1)^k)/3$ , $odd(m)= m\!\mod 2$, and \Dchild$(\phi_i,n)$ returns the highest index less than $n$ of a grandchild of $\phi_i$ in the weak heap (i.\,e, \Dchild$(\phi_i,n)=$ index of the bottommost element in the weak heap which has $\Dancestor{}= \phi_i$ and index $<n$).
\label{xxmerge}}
\end{algorithm}

\begin{theorem}[Average Case of {\sc MergeInsertion}]\label{thm:avgMergeIns}
The sorting algorithm {\sc MergeInsertion} needs $n\log n - c(n)\cdot
n + \Oh(\log n)$ comparisons on the average, where $c(n) \geq
1.3999$.
\end{theorem}

\begin{corollary}[{\sc QuickMergesort} with Base Case {\sc MergeInsertion}]\label{thm:QMSbaseMI}
When using {\sc MergeInsertion} as base case, {\sc
  QuickMergesort} needs at most $n \log n - 1.3999 n + o(n)$ comparisons
and $\Oh(n \log n)$ other instructions on the average.
\end{corollary}

\begin{proof}[Proof of \prettyref{thm:avgMergeIns}]
According to Knuth \cite{Knu73}, {\sc MergeInsertion} requires at most $W(n) =
n \log n - (3- \log 3) n + n (y +1 - 2^y) +\Oh(\log n)$ comparisons
in the worst case, where $y =y(n)= \ceil{\log(3n/4)} - \log(3n/4) \in [0,1)$. 
In the following we want to analyze the average savings relative to the worst case. 
%
Therefore, let $F(n)$ denote the average number of comparisons of the insertion steps of {\sc MergeInsertion}, i.\,e., all comparisons minus the efforts for the weak heap construction, which always takes place. 
Then, we obtain the recurrence relation
\begin{eqnarray*}
F(n) &=& F(\floor{n/2}) + G(\ceil{n/2})\mbox{, with} \\
G(m) &=& (k_m - \alpha_m)\cdot (m - t_{k_m-1}) + \sum_{j=1}^{k_m-1} j \cdot (t_j - t_{j-1}),
\end{eqnarray*}
with $k_m$ such that $t_{k_m-1} \le m < t_{k_m}$ and some $\alpha_m \in [0,1]$. As we do not analyze the improved version of the algorithm, the insertion of elements with index less or equal $t_{k_m-1}$ requires always the same number of comparisons. Thus, the term $\sum_{j=1}^{k_m-1} j \cdot (t_j - t_{j-1})$ is independent of
the data. However, inserting an element after $t_{k_m-1}$ may either need $k_m$ or $k_m-1$ comparisons. This is where $\alpha_m$ comes from. Note that $\alpha_m$ only depends on $m$.
We split $F(n)$ into $F'(n) + F''(n)$ with 
\begin{align*}
&&F'(n) &= F'(\floor{n/2}) + G'(\ceil{n/2}) &&\mbox{and} \\
&&G'(m) &= (k_m - \alpha_m)\cdot (m - t_{k_m-1})&&\mbox{with $k_m$ such that $t_{k_m-1} \le m < t_{k_m}$},
\intertext{and}
&&F''(n) &= F''(\floor{n/2}) + G''(\ceil{n/2})&&\mbox{and} \\
&&G''(m) &= \sum_{j=1}^{k_m-1} j \cdot (t_j - t_{j-1}) &&\mbox{with $k_m$ such that $t_{k_m-1} \le m <  t_{k_m}$}.
\end{align*}

For the average case analysis, we have that $F''(n)$ is independent of
the data. For $n = (4/3)\cdot 2^k$ we have $G'(n)=0$, and hence, $F'(n)=0$. Since otherwise $G'(n)$ is non-negative, this proves that exactly for $n = (4/3)\cdot 2^k$ the average case matches the worst case.

Now, we have to estimate $F'(n)$ for arbitrary $n$.
We have to consider the
calls to binary insertion more closely.
To insert a new element into an array of $m-1$ elements either needs $\ceil{\log m}-1$ or $\ceil{\log m}$ comparisons. For a moment assume that the element is inserted at every position with the same probability. Under this assumption the analysis in the proof of \prettyref{prop:avgIns} is valid, which states that
 \begin{align*}
 T_{\mathrm{InsAvg}}(m) 
  &= \ceil{\log m } + 1 - \frac{2^{\ceil{\log m} }}{m}
 \end{align*}
 comparisons are needed on average.

The problem is that in our case the probability at which position an element is inserted is not uniformly distributed. However, it is monotonically increasing with the index in the array (indices as in our implementation). 
Informally speaking, this is because if an element is inserted further to the right, then for the following elements there are more possibilities to be inserted than if the element is inserted on the left.

Now, \Binaryinsert{} can be implemented such that for an odd number of positions the next comparison is made such that the larger half of the array is the one containing the positions with lower probabilities. (In our case, this is the part with the lower
indices~-- see \prettyref{fig:binarysearch}.)
That means the less probable positions lie on rather longer paths in the search tree, and hence, the average path length is better than in the uniform case.
%
Therefore, 
we may assume a uniform distribution in the following as an upper bound.

 In each
of the recursion steps we have $\ceil{n/2}-t_{k_{\ceil{n/2}}-1}$ calls to binary
insertion into sets of size $\ceil{n/2}+t_{k_{\ceil{n/2}}-1}-1$ elements each.
Hence, for inserting one element, the difference to the worst case is $\frac{2^{\ceil{\log \ceil{n/2}+t_{k_{\ceil{n/2}}-1}}}}{\ceil{n/2}+t_{k_{\ceil{n/2}}-1}} - 1$.
Summing up, we obtain for the average savings $S(n) = W(n) - (F(n)+ \text{\it weak-heap-construction}(n))$ w.\,r.\,t.\ the worst case number $W(n)$ the recurrence
$$S(n) \geq S(\floor{n/2}) + (\ceil{n/2} -t_{k_{\ceil{n/2}}-1}) \cdot \left(\frac{2^{\ceil{\log (\ceil{n/2}+t_{k_{\ceil{n/2}}-1})}}}{\ceil{n/2}+t_{k_{\ceil{n/2}}-1}} - 1\right).$$
For $m\in \R_{>0}$ we write $m = 2^{\ell_m -  \log 3 + x}$ with $x \in [0,1)$ and we set
$$f(m) = ( m - 2^{\ell_m -  \log 3})\cdot\left(\frac{2^{\ell_m}}{m + 2^{\ell_m -  \log 3}} -1\right).$$
Recall that we have $t_k = (2^{k+1}+(-1)^k)/3$. Thus, $k_m$ and $\ell_m$ coincide for most $m$ and differ by at most 1 for a few values where $m$ is close to $t_{k_m}$ or $t_{k_m-1}$.
Since 
in both cases $f(m)$ is smaller than some constant, 
this implies that $f(m)$ and $ ( m - t_{k_m-1})\cdot\left(\frac{2^{\ceil{\log( m+t_{k_m-1})}}}{m+t_{k_m-1}} - 1\right)$ differ by at most a constant. Furthermore, $f(m)$ and $f(m+1/2)$ differ by at most a constant. Hence, we have:
$$S(n) \geq S(n/2) + f(n/2) +  \Oh(1).$$
Since we have $f(n/2) = f(n) /2$, this resolves to 
$$S(n) \geq \sum_{i>0} f(n/2^i)+ \Oh(\log n)=  \sum_{i>0}  f(n) / 2^i + \Oh(\log n)= f(n)+ \Oh(\log n).$$
With $n = 2^{k -  \log 3 + x}$ this means up to $\Oh(\log n /n)$-terms
\begin{align*}
\frac{S(n)}{n} &\approx  \frac{ 2^{k -  \log 3 + x} - 2^{k -  \log 3}}{2^{k -  \log 3 + x}}\cdot\left(\frac{2^{k}}{2^{k -  \log 3 + x} + 2^{k -  \log 3}} -1\right)\\
    &=(1-2^{-x})\cdot \left(\frac{3 }{2^{  x} +1} -1\right).
\end{align*}
Writing $F(n)= n\log n - c(n)\cdot
n + \Oh(\log n)$ we obtain with \cite{Knu73}
$$ c(n) \geq - (F(n) - n \log n )/ n = (3- \log 3)- (y +1 - 2^y)    + S(n)/n,$$
where $y = \ceil{\log(3n/4)} - \log(3n/4) \in [0,1)$, i.\,e., $n = 2^{\ell - \log 3-y}$ for some $\ell \in \Z$. With $y = 1-x$ it follows
$$c(n) \geq (3- \log 3) - (1-x +1 - 2^{1-x}) + (1-2^{-x})\cdot\left(\frac{3 }{2^{  x} +1} -1\right) > 1.3999.$$
This function reaches its minimum in $[0,1)$ for $x = \log\left(\ln 8 -1+\sqrt{(1-\ln 8)^2-1}\right)$. 

It is not difficult to observe that $c(2^k) = 1.4$.
For the factor $e(n)$
in $n \log n - e(n) + \Oh(\log n)$ we have
$e(2^k) = 3 - \log 3 + (x +1 -2^x)$, where
$x = \ceil {\log (3/4) \cdot 2^k} - \log ((3/4) \cdot 2^k)$.
We know that $x$ can be rewritten as 
$x = \ceil{\log (3) + \log(2^k/4)} - (\log 3 + \log(2^k/4)
   = \ceil{\log 3} - \log 3 = 2 - \log 3$.
Hence, we have $e(n) = - 3 \log (3) + (3 \log (3) - 2 ^{2-\log(3)}) = -
4/3$.  Finally, we are interested in the value $W(n)-S(n)
 = W(2^k)-S(2^k) = -4/3 - 1/15 = -1.4$.
\end{proof}

\section{Experiments}
Our experiments consist of two parts. First, we compare the different algorithms we use as base cases, i.\,e., {\sc MergeInsertion}, its improved variant, and {\sc Insertionsort}. The results can be seen in \prettyref{fig:base_cases}. Depending on the size of the arrays the displayed numbers are averages over 10-10000 runs\footnote{Our experiments were run on one core of
an Intel Core i7-3770 CPU (3.40GHz, 8MB Cache) with 32GB RAM; Operating system:
Ubuntu Linux 64bit; Compiler: GNU's \texttt{g++} (version 4.6.3) optimized with flag
\texttt{-O3}.}. The data elements we sorted were randomly chosen 
64-bit integers\footnote{To rely on objects being
handled we avoided the flattening of the array structure 
by the compiler. Hence, for the running time experiments, 
and in each comparison taken, we left the counter 
increase operation intact.}.

The outcome in \prettyref{fig:base_cases} shows that our improved {\sc
  MergeInsertion} implementation achieves results for the
 constant $\kappa$ of the linear term in the range
of $[-1.43,-1.41]$ (for some values of $n$ are even smaller than
$-1.43$). Moreover, the standard implementation with slightly more comparisons
is faster than {\sc Insertionsort}.  By the $\Oh(n^2)$ work, the
resulting runtimes for all three implementations raises quickly,
so that only moderate values of $n$ can be handled.

\begin{figure}[ht]
\centerline{
\includegraphics[width=6.5cm,height=6.5cm]{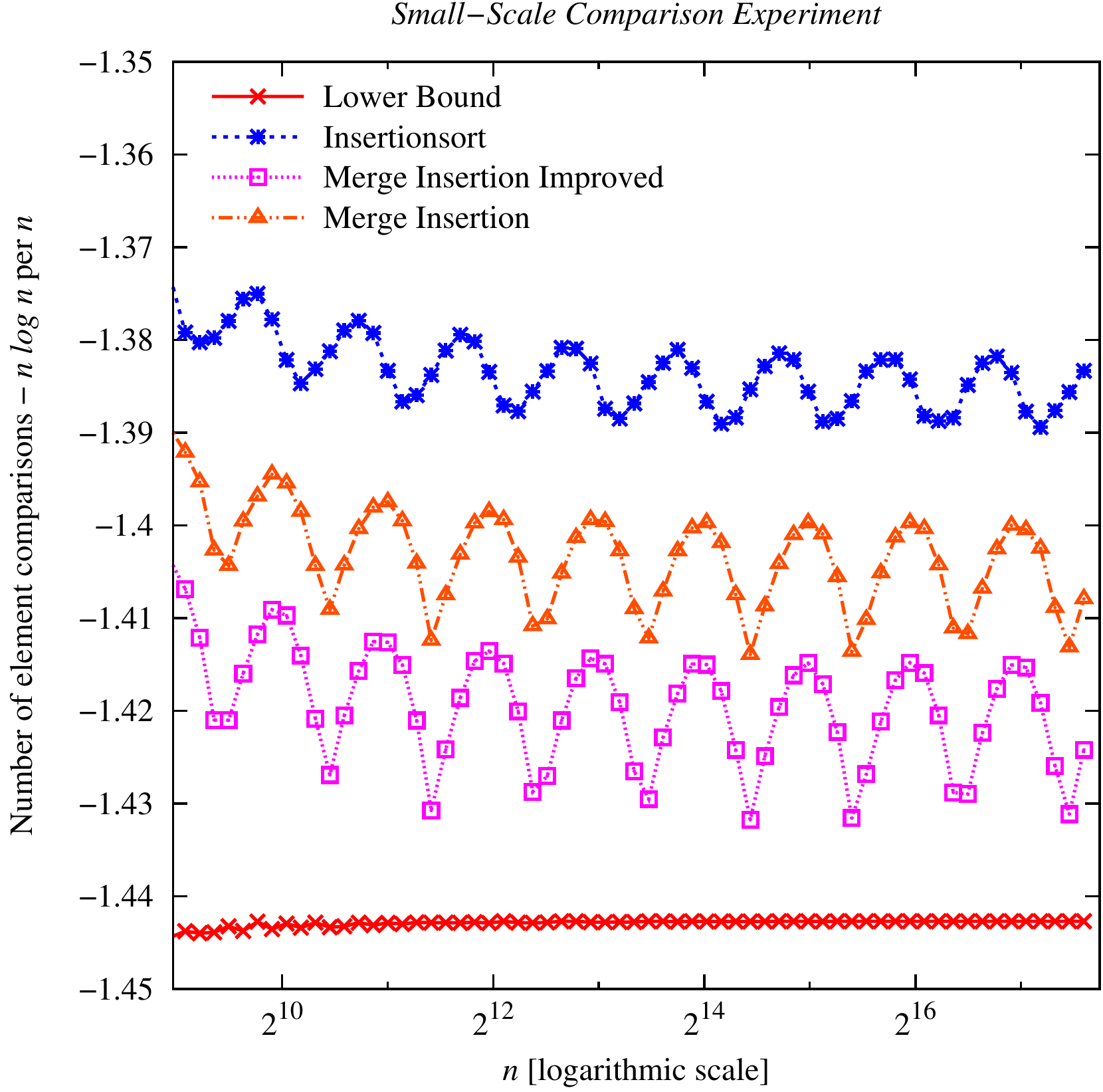}\hspace{3mm}
\includegraphics[width=6.5cm,height=6.5cm]{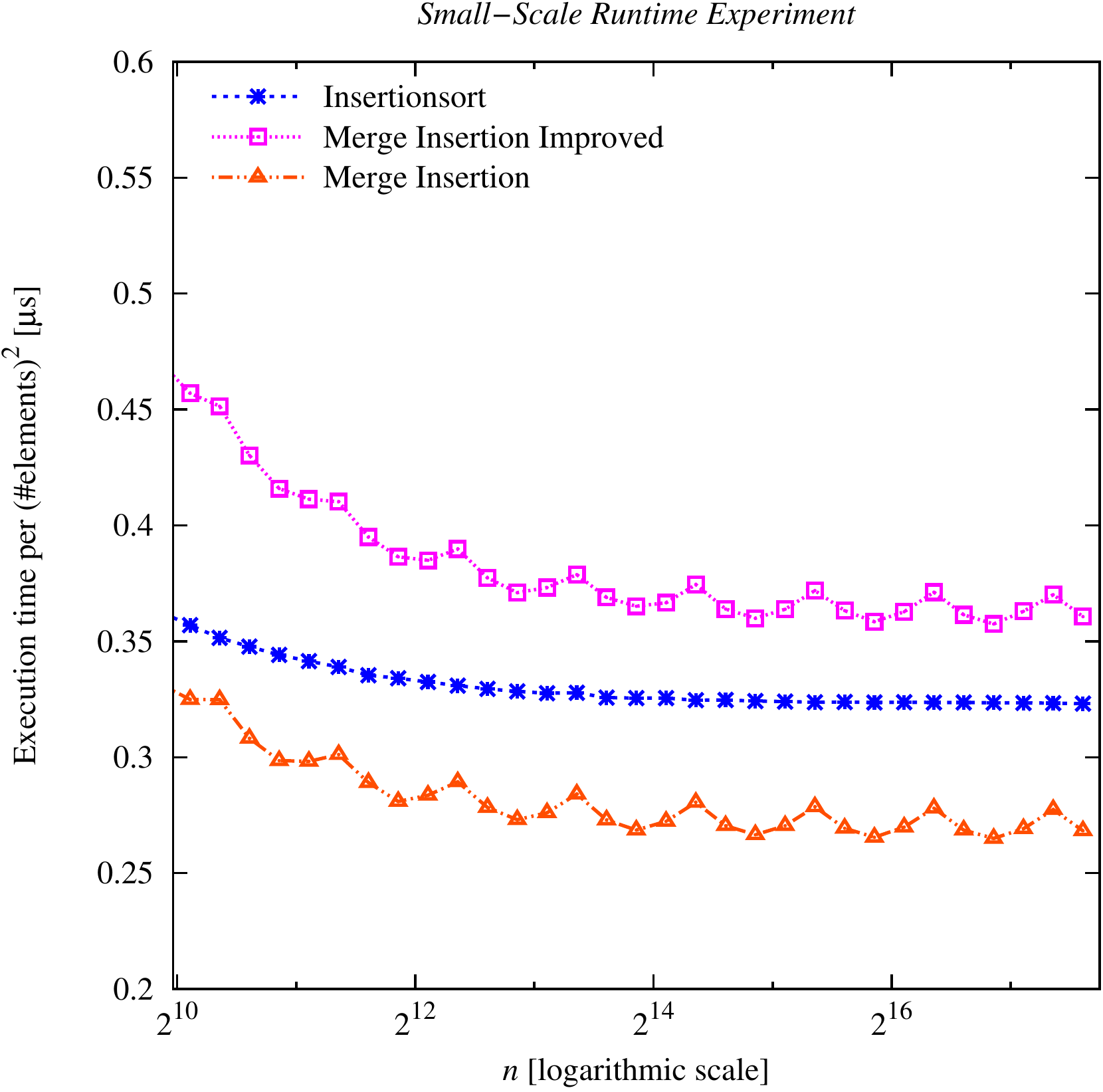}
}
\caption{Comparison of {\sc MergeInsertion}, its improved variant and {\sc Insertionsort}. For the number of comparisons  $n\log n +\kappa n$ the value of $\kappa$ is displayed. }\label{fig:base_cases}
\end{figure}

The second part of our experiments (shown in Fig.~\ref{fig:compare_QMS}) consists of the comparison of {\sc QuickMergesort} (with base cases of constant and growing size)
and {\sc QuickWeakHeapsort} with state-of-the-art algorithms as STL-{\sc Introsort} (i.\,e., {\sc Quicksort}), STL-{\sc stable-sort} (an implementation of {\sc Mergesort}) and {\sc Quicksort} with median of $\sqrt{n}$ elements for pivot selection. For {\sc QuickMergesort} with base cases, the improved variant of {\sc MergeInsertion} is used to sort subarrays of size up to $40 \log_{10} n$. For the normal {\sc QuickMergesort} we used base cases of size $\leq 9$.
We also implemented {\sc QuickMergesort} with median of three for pivot selection, which turns out to be practically efficient, although it needs slightly more comparisons than {\sc QuickMergesort} with median of $\sqrt{n}$. However, since also the larger half of the partitioned array can be sorted with {\sc Mergesort}, the difference to the median of $\sqrt{n}$ version is not as big as in  {\sc QuickHeapsort} \cite{DiekertW13Quick}.
As suggested by the theory, we see that our improved {\sc QuickMergesort} implementation with growing size base cases {\sc MergeInsertion} 
yields a result for the constant in the linear term that is 
in the range of $[-1.41,-1.40]$ -- close to the lower bound. 
 However, for the running time, normal {\sc QuickMergesort} as well as the STL-variants {\sc Introsort} (\texttt{std::sort}) and {\sc BottomUpMergesort} (\texttt{std::stable\_sort}) are slightly better.
With about 15\% the time gap, however, is not overly big, and may be bridged with 
additional efforts like skewed pivots and refined partitioning. Also, if comparisons are more expensive, {\sc QuickMergesort}
should perform significantly faster than {\sc Introsort}.

\begin{figure}[ht]
\centerline{
\includegraphics[width=6.5cm,height=6.5cm]{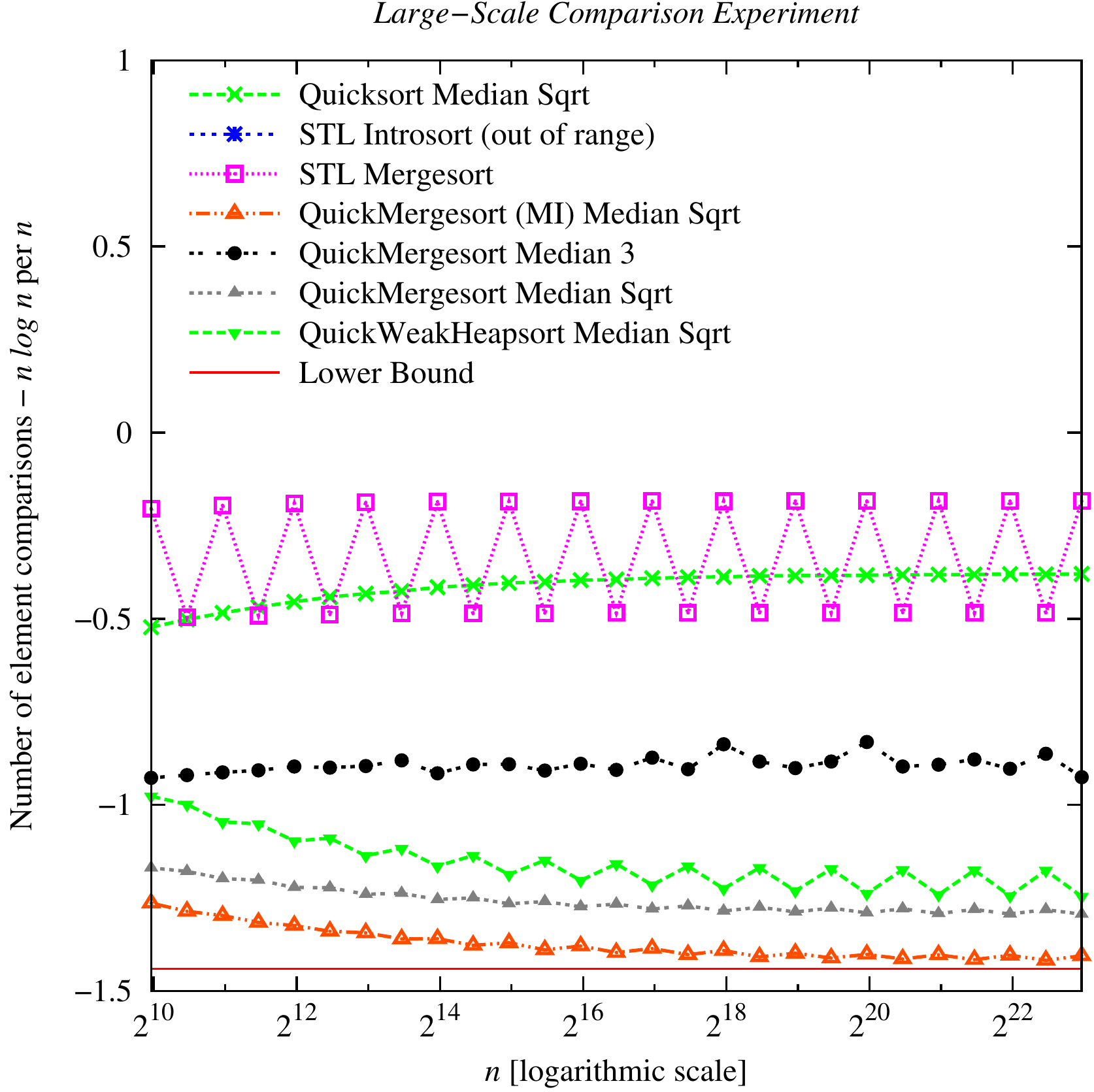}\hspace{3mm}
\includegraphics[width=6.5cm,height=6.5cm]{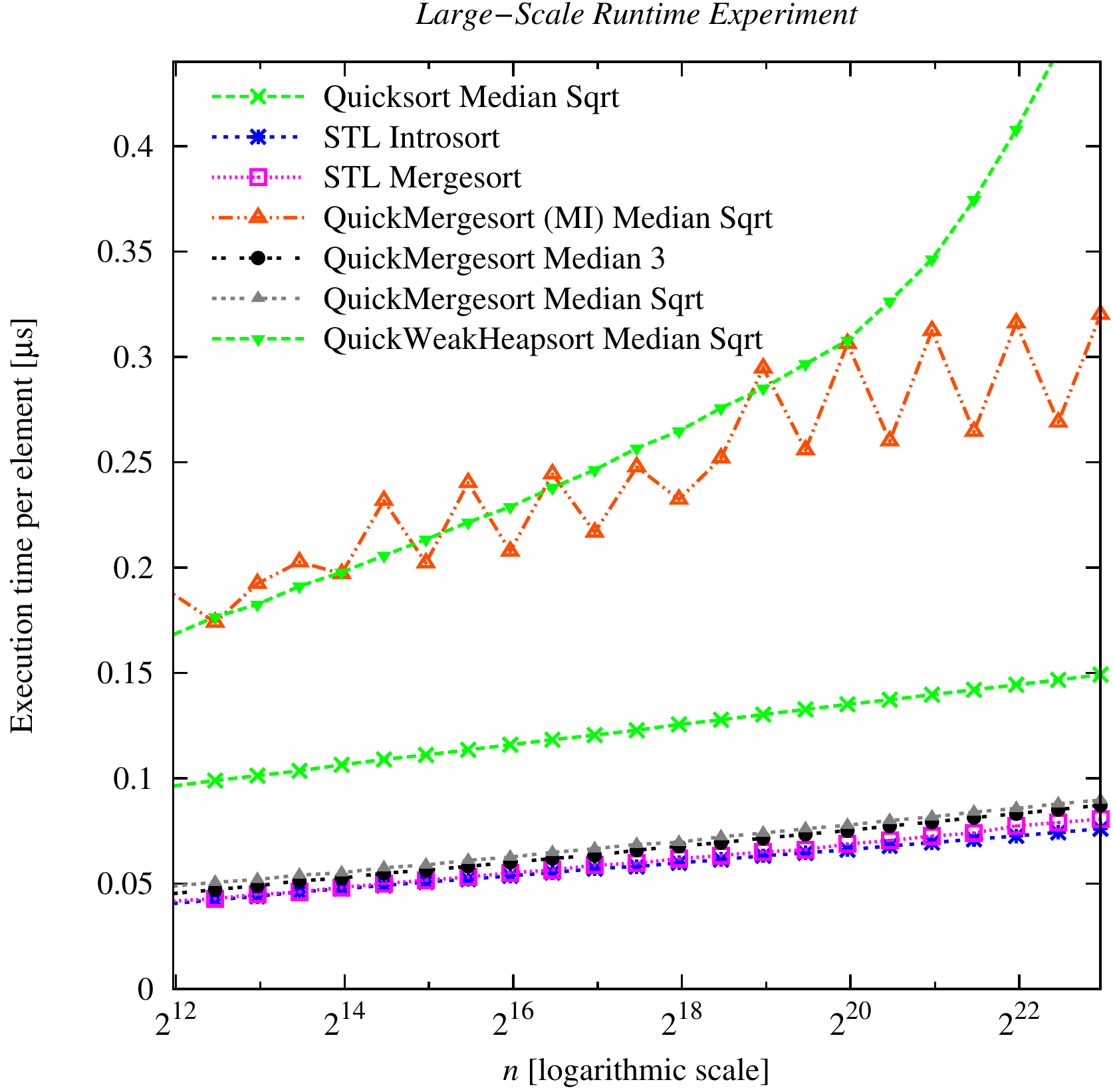}
}
\caption{{Comparison of \sc QuickMergesort} (with base cases of constant and growing size) and {\sc QuickWeakHeapsort} with other sorting algorithms; (MI) is short for 
including growing size base cases derived from {\sc MergeInsertion}. For  the number of comparisons  $n\log n +\kappa n$ the value of $\kappa$ is displayed. }\label{fig:compare_QMS}
\end{figure}

\section{Concluding Remarks}

Sorting $n$ elements remains a fascinating topic for computer
scientists both from a theoretical and from a practical point of
view. 
With {\sc QuickXsort} we have described a procedure how to convert an external sorting algorithm into an internal one introducing only $o(n)$ additional comparisons on average. We presented {\sc QuickWeakHeapsort} and {\sc QuickMergesort} as two examples for this construction.  {\sc QuickMergesort}  is close to the lower bound for the average number of comparisons and at the same time is practically efficient, even when the comparisons are fast.

Using {\sc MergeInsertion} to sort base cases of growing size for {\sc QuickMergesort}, we derive an 
an upper bound of
$n \log n - 1.3999n + o(n)$ comparisons for the average case. As far as we know a better result has not been published before.
Our experimental results validate the theoretical considerations and indicate that the factor $-1.43$ can be beaten. 
Of course, there is still room in closing the gap to the lower bound of
$n \log n - 1.44n + \Oh(\log n)$ comparisons.

\clearpage


\vfill
\pagebreak

\appendix

\section{Pseudocode}
\begin{algorithm}[hbt!]
\SetKwInput{Proc}{procedure}
\SetKwFor{For}{for}{}{}
\SetKwFor{While}{while}{}{}
\SetKw{Return}{return}
\DontPrintSemicolon
\Proc{\Construct{}($s$: array of elements,
  $r$: array of $n$ bits, 
  $m$ bound)}
\For {$k = m-1\ \mbox{\bf downto}\ 1$} {
  \If {$i+1 = k$} {
    \If {$\textit{even}(k)$} {
         $\Link (\Dancestor (i),i)$ \;
      }
      $k \leftarrow \lfloor k / 2 \rfloor$\;
    }
    \Else {
         $\Link (\Dancestor(i),i)$ \;
    }
  }
\;
\Proc{\Dancestor{}($j$: index)}
\While {$(j\mbox{ \emph{\textbf{bitand}} }1) = r_{\lfloor j / 2 \rfloor}$} 
{
     $j \leftarrow \lfloor j / 2 \rfloor$\;   
}
\mbox{\bf return} $\lfloor j / 2 \rfloor$\;
\mbox{ }\;
\;
\Proc{\Link{}($i,j$: indices)} 
    \If {$s_j < s_i$} {
      $\mbox{\em swap}(s_i, s_j)$\; 
      $r_j \leftarrow 1 - r_j$\;
    } 
\;
\Proc{\Dchild{}($i,j$: indices)} 
    $x \leftarrow \mathit{secondchild}(i)$\;
    \While {$\mathit{firstchild}(x) < j$} { 
      $x \leftarrow \mathit{firstchild}(x)$\;
    }
    \Return{$x$}\;
%
\caption{Constructing a weak heap for {\sc MergeInsertion}.
\label{tournament}} 
\end{algorithm}

\begin{algorithm}[hbt!]
\SetKwInput{Proc}{procedure}
\SetKwFor{For}{for}{}{}
\SetKwFor{While}{while}{}{}
\SetKw{Return}{return}
\DontPrintSemicolon
\Proc{\Binaryinsert{}(
  $s$: array of $n$ elements, 
  $\phi$: array of $n$ integers, 
  $r$: array of $n$ bits, 
  $f,d,t$ integers)}

    \For {$j = t\ \mbox{\bf downto}\ f+d+1$} {
         $\textit{swap} (\phi_{j-1},\phi_j)$ \;
    }
    $l \leftarrow f$\;
    $r \leftarrow f+d$\;
    \While {$l < r$} {
      $m \leftarrow (l+r)/2$\;
      \If {$s_{\phi_{f+d}} > s_{\phi_m}$} { 
        $l \leftarrow m+1$\;
      }
      \Else {
        $r \leftarrow m$\;
     }
   }
   \For {$j = f+d\ \mbox{\bf downto}\ l$} {
         $\textit{swap} (\phi_{j-1},\phi_j)$ \;
   }
\caption{Binary insertion of elements in {\sc MergeInsertion} 
algorithm.\label{fig:binarysearch}} 
\end{algorithm}

\begin{algorithm}[hbt!]
\SetKwInput{Proc}{procedure}
\SetKwFor{For}{for}{}{}
\SetKwFor{While}{while}{}{}
\SetKw{Return}{return}
\DontPrintSemicolon

\Proc{\Mergeinsertionrecursive{}(
  $s$: array of $n$ elements, 
  $\phi$: array of $n$ integers, 
  $r$: array of $n$ bits
)} 
\If {$k > 2$} {
\Mergeinsertionrecursive$(k \ \mathit{div}\ 2)$\;
\Merge$(k)$\;
}
\;
\Proc{\Mergeinsertion{}(
  $s$: array of $n$ elements, 
  $\phi$: array of $n$ integers, 
  $r$: array of $n$ bits
)} 
\;
\Construct$(n)$\;
\Mergeinsertionrecursive$(n)$\;
\caption{Main routine and recursive call for {\sc MergeInsertion} 
algorithm.\label{binaryinsertion}} 
\end{algorithm}

\end{document}

%% file: output.pdf_t
\begin{picture}(0,0)%
\includegraphics{output.pdf}%
\end{picture}%
\setlength{\unitlength}{1657sp}%
\begingroup\makeatletter\ifx\SetFigFont\undefined%
\gdef\SetFigFont#1#2#3#4#5{%
  \reset@font\fontsize{#1}{#2pt}%
  \fontfamily{#3}\fontseries{#4}\fontshape{#5}%
  \selectfont}%
\fi\endgroup%
\begin{picture}(3859,3399)(2330,-3910)
\put(5866,-3106){\makebox(0,0)[lb]{\smash{{\SetFigFont{10}{12.0}{\familydefault}{\mddefault}{\updefault}{\color[rgb]{0,0,0}9}%
}}}}
\put(2446,-3796){\makebox(0,0)[lb]{\smash{{\SetFigFont{10}{12.0}{\familydefault}{\mddefault}{\updefault}{\color[rgb]{0,0,0}2}%
}}}}
\put(3511,-2026){\makebox(0,0)[lb]{\smash{{\SetFigFont{7}{8.4}{\familydefault}{\mddefault}{\updefault}{\color[rgb]{0,0,0}2}%
}}}}
\put(4411,-1396){\makebox(0,0)[lb]{\smash{{\SetFigFont{7}{8.4}{\familydefault}{\mddefault}{\updefault}{\color[rgb]{0,0,0}1}%
}}}}
\put(3601,-721){\makebox(0,0)[lb]{\smash{{\SetFigFont{7}{8.4}{\familydefault}{\mddefault}{\updefault}{\color[rgb]{0,0,0}0}%
}}}}
\put(5281,-2446){\makebox(0,0)[lb]{\smash{{\SetFigFont{10}{12.0}{\familydefault}{\mddefault}{\updefault}{\color[rgb]{0,0,0}7}%
}}}}
\put(4711,-3091){\makebox(0,0)[lb]{\smash{{\SetFigFont{10}{12.0}{\familydefault}{\mddefault}{\updefault}{\color[rgb]{0,0,0}5}%
}}}}
\put(4036,-3091){\makebox(0,0)[lb]{\smash{{\SetFigFont{10}{12.0}{\familydefault}{\mddefault}{\updefault}{\color[rgb]{0,0,0}6}%
}}}}
\put(3492,-2413){\makebox(0,0)[lb]{\smash{{\SetFigFont{10}{12.0}{\familydefault}{\mddefault}{\updefault}{\color[rgb]{0,0,0}4}%
}}}}
\put(3541,-1111){\makebox(0,0)[lb]{\smash{{\SetFigFont{10}{12.0}{\familydefault}{\mddefault}{\updefault}{\color[rgb]{0,0,0}1}%
}}}}
\put(5319,-2052){\makebox(0,0)[lb]{\smash{{\SetFigFont{7}{8.4}{\familydefault}{\mddefault}{\updefault}{\color[rgb]{0,0,0}3}%
}}}}
\put(3331,-3796){\makebox(0,0)[lb]{\smash{{\SetFigFont{10}{12.0}{\familydefault}{\mddefault}{\updefault}{\color[rgb]{0,0,0}$11$}%
}}}}
\put(2461,-3417){\makebox(0,0)[lb]{\smash{{\SetFigFont{7}{8.4}{\familydefault}{\mddefault}{\updefault}{\color[rgb]{0,0,0}9}%
}}}}
\put(3455,-3406){\makebox(0,0)[lb]{\smash{{\SetFigFont{7}{8.4}{\familydefault}{\mddefault}{\updefault}{\color[rgb]{0,0,0}8}%
}}}}
\put(2908,-2705){\makebox(0,0)[lb]{\smash{{\SetFigFont{7}{8.4}{\familydefault}{\mddefault}{\updefault}{\color[rgb]{0,0,0}4}%
}}}}
\put(4100,-2712){\makebox(0,0)[lb]{\smash{{\SetFigFont{7}{8.4}{\familydefault}{\mddefault}{\updefault}{\color[rgb]{0,0,0}5}%
}}}}
\put(4737,-2716){\makebox(0,0)[lb]{\smash{{\SetFigFont{7}{8.4}{\familydefault}{\mddefault}{\updefault}{\color[rgb]{0,0,0}7}%
}}}}
\put(5877,-2728){\makebox(0,0)[lb]{\smash{{\SetFigFont{7}{8.4}{\familydefault}{\mddefault}{\updefault}{\color[rgb]{0,0,0}6}%
}}}}
\put(4381,-1779){\makebox(0,0)[lb]{\smash{{\SetFigFont{10}{12.0}{\familydefault}{\mddefault}{\updefault}{\color[rgb]{0,0,0}3}%
}}}}
\put(2912,-3098){\makebox(0,0)[lb]{\smash{{\SetFigFont{10}{12.0}{\familydefault}{\mddefault}{\updefault}{\color[rgb]{0,0,0}8}%
}}}}
\end{picture}%

%% file: quickXsort_arXiv.bbl
\begin{thebibliography}{10}

\bibitem{quickheap}
D.~Cantone and G.~Cinotti.
\newblock Quick{H}eapsort, an efficient mix of classical sorting algorithms.
\newblock {\em Theoretical Comput. Sci.}, 285(1):25--42, 2002.

\bibitem{DiekertW13Quick}
V.~Diekert and A.~Wei{\ss}.
\newblock Quickheapsort: Modifications and improved analysis.
\newblock In A.~A. Bulatov and A.~M. Shur, editors, {\em CSR}, volume 7913 of
  {\em Lecture Notes in Computer Science}, pages 24--35. Springer, 2013.

\bibitem{Dut93}
R.~D. Dutton.
\newblock Weak-heap sort.
\newblock {\em BIT}, 33(3):372--381, 1993.

\bibitem{edelkampstiegeler}
S.~Edelkamp and P.~Stiegeler.
\newblock Implementing {HEAPSORT} with $n \log n - 0.9n$ and {QUICKSORT} with
  $n \log n + 0.2 n$ comparisons.
\newblock {\em ACM Journal of Experimental Algorithmics}, 10(5), 2002.

\bibitem{EW00}
S.~Edelkamp and I.~Wegener.
\newblock On the performance of {W}eak-{H}eapsort.
\newblock In {\em 17th Annual Symposium on Theoretical Aspects of Computer
  Science}, volume 1770, pages 254--266. Springer-Verlag, 2000.

\bibitem{ElmasryKS12}
A.~Elmasry, J.~Katajainen, and M.~Stenmark.
\newblock Branch mispredictions don't affect mergesort.
\newblock In {\em SEA}, pages 160--171, 2012.

\bibitem{FordJ59}
J.~Ford, Lester~R. and S.~M. Johnson.
\newblock A tournament problem.
\newblock {\em The American Mathematical Monthly}, 66(5):pp. 387--389, 1959.

\bibitem{Katajainen98}
J.~Katajainen.
\newblock The {U}ltimate {H}eapsort.
\newblock In {\em CATS}, pages 87--96, 1998.

\bibitem{KatajainenPT96}
J.~Katajainen, T.~Pasanen, and J.~Teuhola.
\newblock Practical in-place mergesort.
\newblock {\em Nord. J. Comput.}, 3(1):27--40, 1996.

\bibitem{Knu73}
D.~E. Knuth.
\newblock {\em Sorting and Searching}, volume~3 of {\em The Art of Computer
  Programming}.
\newblock Addison Wesley Longman, 2nd edition, 1998.

\bibitem{MartinezR01}
C.~Mart\'{\i}nez and S.~Roura.
\newblock Optimal {S}ampling {S}trategies in {Quicksort} and {Quickselect}.
\newblock {\em SIAM J. Comput.}, 31(3):683--705, 2001.

\bibitem{Reinhardt92}
K.~Reinhardt.
\newblock Sorting {\it in-place} with a {\it worst case} complexity of $n \log
  n- 1.3 n + o(\log n)$ comparisons and $\epsilon n \log n + o(1)$ transports.
\newblock In {\em ISAAC}, pages 489--498, 1992.

\bibitem{Weg93}
I.~Wegener.
\newblock {B}ottom-up-{H}eapsort, a new variant of {H}eapsort beating, on an
  average, {Q}uicksort (if $n$ is not very small).
\newblock {\em Theoretical Comput. Sci.}, 118:81--98, 1993.

\end{thebibliography}
